\documentclass{IEEEtran4PSCC}
\ifCLASSINFOpdf
   \usepackage[pdftex]{graphicx}
\else
   \usepackage[dvips]{graphicx}
\fi
\usepackage[cmex10]{amsmath}
\usepackage{amssymb}

\usepackage{algorithm}
\usepackage{algpseudocode}
\usepackage{multirow}
\usepackage[dvipsnames]{xcolor} 
\usepackage[normalem]{ulem}
\usepackage{cuted}  % or \usepackage{sttools} in some templates

\usepackage{amsthm} % for proof environment
          
\newtheorem{Proposition}{Proposition}

\theoremstyle{remark}
\newtheorem*{rmk}{Remark}

\usepackage{graphicx}      
\usepackage[export]{adjustbox}

\usepackage[intoc]{nomencl}
\makenomenclature
\DeclareSymbolFont{bbold}{U}{bbold}{m}{n}
\DeclareSymbolFontAlphabet{\mathbbold}{bbold}

\usepackage{accents}
\newcommand{\ubar}[1]{\underaccent{\bar}{#1}}
\newcommand{\utilde}[1]{\underaccent{\tilde}{#1}}

\usepackage{balance}

\allowdisplaybreaks

\usepackage{mathtools} % use \coloneqq

\newcommand{\edit}[1]{#1}

\makeatletter
\let\old@ps@headings\ps@headings
\let\old@ps@IEEEtitlepagestyle\ps@IEEEtitlepagestyle
\def\psccfooter#1{%
    \def\ps@headings{%
        \old@ps@headings%
        \def\@oddfoot{\strut\hfill#1\hfill\strut}%
        \def\@evenfoot{\strut\hfill#1\hfill\strut}%
    }%
    \def\ps@IEEEtitlepagestyle{%
        \old@ps@IEEEtitlepagestyle%
        \def\@oddfoot{\strut\hfill#1\hfill\strut}%
        \def\@evenfoot{\strut\hfill#1\hfill\strut}%
    }%
    \ps@headings%
}
\makeatother

\psccfooter{%
    \parbox{\textwidth}{\hrulefill\\ \small{\copyright\ 2026 The Author(s). This manuscript version is made available under the CC BY-NC-ND 4.0 license.}}%
}

\begin{document}
\title{Unified Sensitivity-Based Heuristic for Optimal Line Switching and Substation Reconfiguration}

%% To specify the authors when (number of affiliations <= 2)
\author{
\IEEEauthorblockN{Zongqi Hu, Weiqi Meng, and Bai Cui}
\IEEEauthorblockA{Department of Electrical and Computer Engineering \\
Iowa State University\\
Ames, IA, USA\\
\{zqhu, mengwq, baicui\}@iastate.edu}
% \and
% \IEEEauthorblockN{Author n.1 Name per Affiliation B\\ Author n.2 Name per Affiliation B}
% \IEEEauthorblockA{(Affiliation B) Department Name of Organization \\
% Name of the organization, acronyms acceptable\\
% City, Country\\
% \{email author n.1, email author n.2\}@domain (if desired)}
}

%% To specify the authors when (number of affiliations > 2)
% \author{\IEEEauthorblockN{Author n.1\IEEEauthorrefmark{1},
% Author n.2\IEEEauthorrefmark{2},
% Author n.3\IEEEauthorrefmark{3}, 
% Author n.4\IEEEauthorrefmark{3} and
% Author n.5\IEEEauthorrefmark{4}}
% \IEEEauthorblockA{\IEEEauthorrefmark{1} Department Name of Organization A\\
% Name of the organization A,
% Address A\\ Emails if wanted}
% \IEEEauthorblockA{\IEEEauthorrefmark{2} Department Name of Organization B\\
% Name of the organization B,
% Address B\\ Emails if wanted}
% \IEEEauthorblockA{\IEEEauthorrefmark{3} Department Name of Organization C\\
% Name of the organization C,
% Address C\\ Emails if wanted}
% \IEEEauthorblockA{\IEEEauthorrefmark{4}Department Name of Organization D\\
% Name of the organization D,
% Address D\\ Emails if wanted}
% }

\maketitle

\begin{abstract}
Optimal transmission switching (OTS) determines which transmission lines to remove from service to minimize dispatch costs. Unlike topology design, it alters the operational status of operating lines. Sensitivity-based methods, as advanced optimization techniques, select lines whose outage yields a significant cost reduction. However, these methods overlook bus splitting, an effective congestion management strategy that our work incorporates to achieve improved economic gains. In this work, we formulate an optimal transmission reconfiguration (OTR) problem that incorporates both line switching and bus splitting. We develop a novel approach to quantify the sensitivity of the OTR objective to line switching and bus splitting, establish connections between the proposed sensitivity framework and existing heuristic metrics, prove the equivalence between bus splitting and a generalized line switching to enable unified treatment, and provide a simpler derivation of Bus Split Distribution Factor (BSDF). Simulations on nine IEEE test systems spanning 118 to 13,659 buses demonstrate the high effectiveness of our proposed sensitivity method. They also demonstrate that incorporating bus splitting into transmission reconfiguration achieves greater cost savings than line switching alone. The results confirm the economic advantage of this comprehensive approach to transmission system operation.
\end{abstract}

% Use this to place sponsorships
\thanksto{\noindent This work was supported by the Electric Power Research Center (EPRC) at Iowa State University.}

\section{Introduction}
Optimal transmission switching (OTS) has received renewed interest worldwide as one of the grid-enhancing technologies (GETs) enabling cost-effective alternatives to costly system upgrades for congestion management and improved market efficiency. While the industry has a long history of applying transmission switching for reliability purposes such as voltage violation prevention and overload relief, its use for economic optimization is a more recent development \cite{ots}. Despite its potential, solving the OTS problem to global optimality remains computationally challenging. The problem is known to be NP-hard \cite{lehmann2014complexity}, implying that finding a polynomial-time algorithm for the general case is highly unlikely. In addition, OTS is typically subject to stringent time constraints in practical applications, necessitating scalable solution methods. As a result, practical implementations often rely on heuristic or approximation methods to obtain high-quality solutions within a reasonable computation time.

In one heuristic-based approach, line switching decisions are obtained by sequentially solving the problem with each candidate line being allowed to open one at a time \cite{torres2014practical}. Sensitivity-based classification methods in \cite{sensitivity} model line switching by setting the upper and lower bounds of line power flow to be zero, identifying candidate lines whose switching yields maximum cost reduction. A feasible region-based heuristic method is proposed in \cite{crozier2022feasible}, where differences in operational costs between economic dispatch and DCOPF are compared; if these values differ, the non-zero multipliers of constraints are ranked, and the line power flow constraints corresponding to these multipliers are iteratively removed, and the problem is re-solved.
Additionally, integrating heuristic strategies can enhance the computational performance of mixed-integer linear programming (MILP) solvers; for example, the authors of \cite{hinneck2022optimal} propose a parallel algorithmic architecture for solving MILP-based OTS. \edit{In \cite{marot2025superposition}, an overall framework including line outage and line reconnection is considered by utilizing distribution factors, and extended superposition theorem is applied to indicate the change of line power flow with respect to the system topology variable. However, the resulting MILP OTS problem is still computationally challenging to solve.}

For the approximate methods, the approximation scheme in \cite{jabarnejad2018approximate} maintains the accuracy of the operational cost of the OTS formulation \cite{ots} while significantly reducing computational overhead. Recently, machine-learning-based approaches for OTS \edit{and} substation reconfiguration have also been explored \cite{kim2025dispatch,meng2025flow}. Under uncertainties of renewable energies, a two-stage robust problem is formulated in \cite{zhou2022distributionally} to limit constraint violations and is solved by \edit{approximation} method. To further reduce computational burden, a graph-based method is proposed in \cite{aguilar2025graph}, combining graph-theoretic bounds on voltage angle differences with an improved feasible-solution search. 

For substation reconfiguration (also referred to as bus split), it is taken into account for the optimal reconfiguration of power transmission systems in \cite{heidarifar2014optimal} to reduce congestion and meet $N\!-\!1$ security constraints. An improved node-breaker representation is given in \cite{heidarifar2015network}, using double bus-bar models where connections to the original buses are represented by binary variables. However, these works consider a limited number of switching actions and cannot represent actions for a real bus station. This limitation is addressed through formulations that permit each grid component to connect to either split bus bar \cite{hinneck2021optimal,bastianel2024optimal}; however, these approaches do not incorporate line switching flexibility. \edit{A shift factor framework is introduced in \cite{goldis2016shift} to model breakers in the power flow equations. The OPF problem is then formulated as an MILP, the computation of which becomes time consuming as the system size and number of buses to be split scales. Reference \cite{babaeinejadsarookolaee2021tractable} provides a sensitivity-based method to calculate line power flow with respect to the node breaker position, and the selected bus to be split is based on the sensitivity factor to the line power flow, while neglecting the relationship between bus split and the optimal cost change. This work led to the development of  ``line load reduction'' (LLR) heuristic method in \cite{babaeinejadsarookolaee2023transmission}, where the sum of squared line currents with weights and voltage regulation term are used in the objective function.}

Existing heuristic or \edit{approximation} methods that select line switching address computational issues but generally do not consider line switching and bus split simultaneously. In this paper, we propose a novel sensitivity-analysis-based heuristic algorithm for unified \edit{optimal transmission reconfiguration (OTR)} identification that includes both line switching and substation reconfiguration (bus split). The contributions of this paper are fivefold:
\begin{enumerate}
    \item We develop a novel approach to quantify the sensitivity of the optimal objective value of the \edit{OTR} problem to line switching, which differs from classical sensitivity methods in the literature, such as \cite{sensitivity}.
    \item We establish a formal connection between the proposed sensitivity analysis framework and heuristic metrics commonly used in \cite{sensitivity}, providing theoretical justification for some metrics and motivating improved alternatives.
    \item We demonstrate the relationship between bus split (i.e., substation reconfiguration) and generalized line switching, enabling a unified treatment of bus split and line switching in the sensitivity analysis framework.
    \item We develop novel heuristic algorithms for OTR solution identification based on the new sensitivity analysis framework.
    \item As a byproduct, we propose a significantly simpler derivation of the bus split distribution factor (BSDF)—a recently proposed factor quantifying the sensitivity of line flows to bus split \cite{van2023bus}.
\end{enumerate}

\section{Background}
Consider a power system modeled as a graph $\mathcal{G} = (\mathcal{V},\mathcal{E})$ where $\mathcal{V}$ is the vertex set and $\mathcal{E} \subseteq \mathcal{V}\times \mathcal{V}$ is the (undirected) edge set. We may assume the system has $n$ vertices (buses) and $m$ edges (lines), that is, $|\mathcal{V}| = n$ and $|\mathcal{E}| = m$. The susceptance matrix $B^{0} \in \mathbb{R}^{n\times n}$ is the weighted graph Laplacian that describes the connectivity and weights of the graph. The diagonal elements of $B^{0}$ are the sums of the negative line susceptances (which are positive for inductive lines) incident to the corresponding buses, and the off-diagonal entries are the line susceptances between the two buses corresponding to the row and column indices. Specifically, the elements of the susceptance matrix are
\begin{equation*}
    B^{0}_{ij} = \begin{cases}
        \sum_{k=1}^n b_{ik} & \text{if } i= j \\
        -b_{ij} &\text{if } i \ne j \text{ and } (i,j) \in \mathcal{E} \\
        0 &\text{otherwise}
    \end{cases}
\end{equation*}
where $b_{ij}$ is the negative line susceptance between buses $i$ and $j$. \edit{We find it helpful to further denote the bus-branch incidence matrix as $A^0 \in \{1,-1,0\}^{n\times m}$ and the diagonal matrix of negative line susceptance as $D\in\mathbb{R}^{m\times m}$. It follows that $B^0 = A^0 D (A^0)^\top$}.

\edit{The vector $\theta \in \mathbb{R}^n$ collects the vector of voltage phase angles.} Similarly, let $p^g, p^d \in \mathbb{R}^n_{\ge0}$ be the vectors of generation and load at \edit{all }buses\edit{, and $p \coloneqq p^g - p^d$ is the vector of nodal net injections.} We adopt the DC power flow model in this paper, in which the phase angles and nodal injections are linearly related as
\begin{equation}
    p = B^0 \theta.
\end{equation}

\edit{Let the vector of line flows be defined as $f \coloneqq D(A^0)^\top \theta$, we have 
\begin{equation} \label{eq:flow}
    f = \Psi p
\end{equation}
where the power transfer distribution factor (PTDF) matrix $\Psi$ is given by $\Psi = D(A^0)^\top (B^0)^\dagger$. The matrix $(B^0)^\dagger$ is the Moore-Penrose pseudoinverse of $B^0$.}

% By defining the vector of line flows as $f \coloneqq D(A^0)^\top\theta\textcolor{blue}{=H\theta}$, \textcolor{blue}{where $H\in \mathbb{R}^{m\times n}$ is the branch-by-node matrix}. Therefore, we have
% \begin{equation}
%     f = H (B^0)^{-1}p
% \end{equation}
% \textcolor{blue}{where $H (B^0)^{-1}$ is the PTDF matrix $\Psi$. However, $B^0$ is not invertible, the row and column corresponding to slack bus need to be removed. So reduced matrix $H_{red}\in \mathbb{R}^{m\times (n-1)}$ and $B^0_{red}\in \mathbb{R}^{(n-1)\times (n-1)}$ are introduced where the row and column related to slack bus are eliminated. Thus, PTDF matrix $\Psi\in \mathbb{R}^{m\times n}$ is calculated as:}\\
% \textcolor{blue}{\begin{equation}
%     \Psi_{l,n}=
%     \begin{cases}
%    (H_{red}(B^0_{red})^{-1})_{l,n} & \text{if }n \text{ is not slack bus}\\
%     0 & \text{otherwise} 
%     \end{cases}
% \end{equation}}

The (DC)OTS problem can be formulated as an MILP \cite{ots}, 
% \begin{align}
%     C = \text{min }&c^{T}P\\
%     \text{s.t. }&\theta^{min}_{n}\le \theta_{n}\le \theta^{max}_{n}\\
%     &P^{min}_{ng}\le P_{ng}\le P^{max}_{ng} \text{ for all $g$ and $n$ }\\
%     &P^{min}_{nk}z_{k}\le P_{nk}\le P^{max}_{nk}z_{k} \text{ for all $k$ and $n$ }\\
%     &-\displaystyle\sum_{k}P_{nk}-\displaystyle\sum_{g}P_{ng}-\displaystyle\sum_{d}P_{nd} = 0\\
%     &B_{k}(\theta_{n}-\theta_{m})-P_{nk}+(1-z_{k})M \ge 0\\
%     &B_{k}(\theta_{n}-\theta_{m})-P_{nk}-(1-z_{k})M \ge 0\\
%     &\displaystyle\sum_{k}(1-z_{k})=1
% \end{align}
which \edit{can be computationally demanding especially for real-time operation applications}. For simplicity, we assume \edit{linear cost curve $c^\top p^g$ for the problem}. Recently, bus split has been incorporated in addition to line switching \cite{zhou2021substation}. \edit{They assume that the newly-generated bus after bus split is connected with only one bus and the bus split can be equivalently regarded as retaining the topology but with some power transfer between the split bus and the connected one. Consider line $\ell=(i,j)$, the power transfer scenarios corresponding to split of bus $i$ include generation only, load only, and both. The following matrix is used to express the above cases: $\Delta_{l,i}=(e_i-e_j)\left[p^d_i,\, -p^g_i, \, p^d_i-p^g_i \right] \in \mathbb{R}^{n\times 3}$, where $e_i$ is the $i$-th standard basis in $\mathbb{R}^n$. Similarly, the three power transfer situations of split of bus $j$ can also be represented as: $\Delta_{l,j}=(e_j-e_i)\left[p^d_j, \, -p^g_j, \, p^d_j-p^g_j\right]$.} \edit{In this way, the OTR problem can be formulated as the following MILP problem:}
\begin{subequations} \label{eq:otr}
\begin{align}
    C = \min \quad &c^{\top}p^g\label{exact objective function}\\
    \text{over} \quad & \theta \in \mathbb{R}^n, p^g \in \mathbb{R}^{n},  f \in \mathbb{R}^m, z \in \{0, 1\}^m, \notag\\
    & w_{l,i}\in \{0, 1\}^{3}, w_{l,j}\in \{0, 1\}^{3}, \forall\, \ell = (i,j) \in \mathcal{E} \notag\\
    \text{s.t.} \quad & \ubar{\theta}_i \le \theta_i \le \bar{\theta}_i \\
    & \ubar{p}^g_i \le p^g_i \le \bar{p}^g_i\\
    &\ubar{f}_{\ell}z_{\ell} \le f_\ell \le \bar{f}_\ell z_\ell\\
    &b_{ij}(\theta_{i}-\theta_{j})-f_\ell+(1-z_\ell)M_\ell \ge 0  \notag\\
    &b_{ij}(\theta_{i}-\theta_{j})-f_\ell-(1-z_\ell)M_\ell \le 0 \label{big-M}\\
    &\sum_{\ell\in \mathcal{E}}(1-z_\ell)=1 \\
    &\mathbbold{1}^\top w_{\ell,i}+\mathbbold{1}^\top w_{\ell,j}\le 1-z_\ell \label{line open or bus split}\\
    & \sum_{\ell:i\in \mathcal{N}_{\ell}}\mathbbold{1}^\top w_{\ell,i}\le 1 \label{one power transfer}\\
    &\Psi f = p^g- p^d \notag\\
    &+\sum_{\ell=(i,j)}\Delta_{\ell,i}w_{\ell,i}+\displaystyle\sum_{\ell=(i,j)}\Delta_{\ell,j}w_{\ell,j} \label{new power balance}\\
    &\ubar{f}_\ell \mathbbold{1}\le \Delta_{\ell,i}w_{\ell,i}+\Delta_{\ell,j}w_{\ell,j}\le \bar{f}_\ell \mathbbold{1} \label{new line power flow limit}
\end{align}
\end{subequations}
where $\ubar{f}$ \edit{and $\bar{f}$} are the lower \edit{and upper} bounds of line power flow. The binary variables $z$ are introduced to represent the statuses of lines. When one line is open, the constraints are guaranteed by imposing \edit{big $M_\ell$ for all $\ell \in \mathcal{E}$} in constraint \eqref{big-M} for line $\ell$. $\mathcal{N}_\ell$ is the set of incident buses for line $\ell$. $w_{\ell,i}, w_{\ell,j}\edit{\in \mathbb{R}^3}$ are two vectors of binary variables \edit{with only one element being 1 }that are used to choose power transfer type when bus $i$ or $j$ is split. \edit{For example, Fig. \ref{fig:1} illustrates a situation where bus $i$ is split into $i$ and $i'$, and the load will connect with the newly-generated bus $i'$. Furthermore, the post-split system is shown in Fig. \ref{fig:2} with bus $i'$ and $j$ being merged. So, the load is connected with bus $j$. In this scenario, $w_{\ell,i}=[1, 0, 0]^\top$ and $w_{\ell,j}=[0,0,0]^\top$ with $\ell=(i,j)$. When $z_\ell=0$ and $\mathbbold{1}^\top w_{\ell,i}+\mathbbold{1}^\top w_{\ell,j}=0$, it is simply a line switching. If $z_\ell=0$ and $\mathbbold{1}^\top w_{\ell,i}+\mathbbold{1}^\top w_{\ell,j}=1$, one power transfer scenario is chosen after line switching, making it equivalent to bus split. Constraint \eqref{one power transfer} enforces only one power transfer case to occur. Constraints \eqref{new power balance} and \eqref{new line power flow limit} ensures the power balance after power transfer and the transferred power does not violate the limit of line $\ell=(i, j)$}. In contrast to this problem, the goal of the current paper is to propose an algorithm that can efficiently search for high-quality solutions of \edit{OTR} through sensitivity analysis. By doing so, we are able to pick high-quality switching actions one at a time in order to reduce the dispatch cost while maintaining high computational efficiency.

\begin{figure}[!t]
    \centering
    \includegraphics[width=0.8\linewidth]{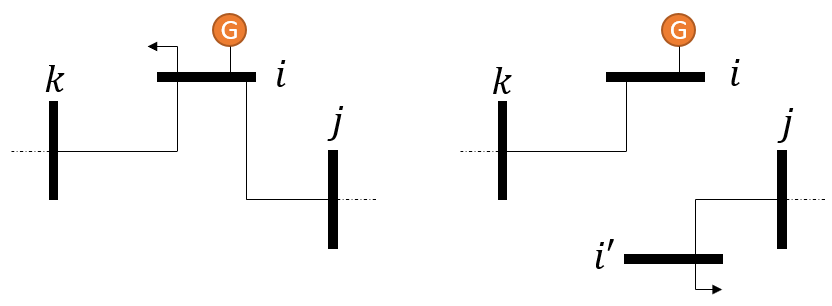}
    \caption{(Left) Original bus-branch model and (right) model with a bus split at bus $i$.}
    \label{fig:1}
\end{figure}

\begin{figure}[!t]
    \centering
    \includegraphics[width=0.8\linewidth]{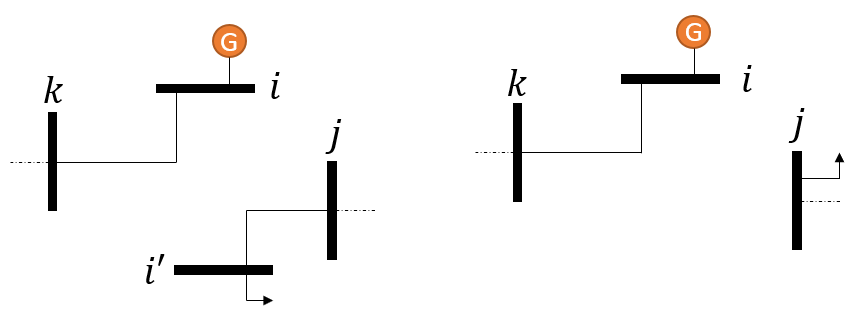}
    \caption{(Left) Post-split bus-branch model and (right) equivalent reduced bus-branch model.}
    \label{fig:2}
\end{figure}

\section{\edit{Sensitivities of OTR Optimal Cost to Parameter Changes}}

In this section, the sensitivity of the optimal objective value with respect to the switching action is derived. By modeling line switching as a discrete change in line susceptance, we can establish the sensitivity by characterizing the sensitivity of the optimal objective value of the OPF problem with respect to line susceptance, which can in turn be derived by applying the \edit{envelope theorem} \cite{Varian1992}. While sensitivity analysis for OTS has been studied previously \cite{sensitivity}, our results differ in both modeling and end result. By formulating the switching action as a change in line \edit{susceptance} rather than as bounds on line flow, the sensitivity metric we obtain aligns more closely with observed system behavior. \edit{In addition, we will also present the sensitivity of the optimal cost with respect to other parameters relevant in modeling bus split.}

\subsection{\edit{Line Switching Sensitivity and First-Order Cost Change}}

% The DC-OTS problem is given as
% \begin{align}
%     C = \text{min } &c^{T}P \label{cost function I}\\
%     \text{s.t. } &B^{0}\theta = P-L \label{power balance}\\
%     &\underline{f}_{i,j} \leq b_{i,j}(\theta_{i}-\theta_{j}) \leq \overline{f}_{i,j} \label{power flow limits}\\
%     &\underline{P} \leq P \leq \overline{P}\label{generation limits}
% \end{align}

When only a single switching action is allowed, the optimal solution can be found by enumerating all currently closed lines (for simplicity, we consider only line openings as actionable line operations). We consider the following DC-OPF problem parametrized by the susceptance matrix $B^0$:
\begin{subequations} \label{eq:OPF}
\begin{align} 
    v(B^0) = \text{min} \quad & c^\top p^g \label{eq:OPF:obj} \\
    \text{s.t.} \quad & B^{0}\theta = p^g - p^d && (\lambda) \label{eq:OPF:a} \\
    & \ubar{f}_{ij} \leq b_{ij}(\theta_{i}-\theta_{j}) \leq \bar{f}_{ij} && (\ubar{\mu}, \bar{\mu}) \label{eq:OPF:b} \\
    & \ubar{p}^g \leq p^g \leq \bar{p}^g && (\ubar{\nu}, \bar{\nu}) \label{eq:OPF:c}
\end{align}
\end{subequations}
where $v(\cdot)$ is the value function of problem \eqref{eq:OPF}; that is, it is the optimal value attained by \eqref{eq:OPF} when the network susceptance matrix is $B^0$. The dual variables are $\lambda \in \mathbb{R}^n$, $\ubar{\mu}, \bar{\mu} \in \mathbb{R}^m_{\ge 0}$, and $\ubar{\nu}, \bar{\nu} \in \mathbb{R}^n_{\ge 0}$. The Lagrangian of problem \eqref{eq:OPF} is given by
\begin{multline} \label{eq:Lagrangian}
    \mathcal{L}(p^g, \theta, \lambda, \ubar{\mu}, \bar{\mu}, \ubar{\nu}, \bar{\nu}; B^0) = c^\top p^g + \lambda^\top (B^{0}\theta - p^g + p^d) \\
    + \sum\nolimits_{(i,j) \in \mathcal{E}} \left( (\bar{\mu}_{ij} - \ubar{\mu}_{ij})  b_{ij}(\theta_i - \theta_j) + \ubar{\mu}_{ij}\ubar{f}_{ij} - \bar{\mu}_{ij}\bar{f}_{ij} \right) \\
    + (\bar{\nu} - \ubar{\nu})^\top p^g + \ubar{\nu}^\top\ubar{p}^g - \bar{\nu}^\top\bar{p}^g .
\end{multline}

We are interested in finding the sensitivity of $v(B^0)$ with respect to  $b_{ij}$ for some line $(i,j)\in\mathcal{E}$, which can be obtained using the envelope theorem \cite[Sect. 27.4]{Varian1992}. The theorem states that, given an optimal primal-dual pair to problem \eqref{eq:OPF} as $x^* \coloneqq ((p^g)^*, \theta^*, \lambda^*, \ubar{\mu}^*, \bar{\mu}^*, \ubar{\nu}^*, \bar{\nu}^*)$, and assuming the differentiability of $v(\cdot)$ in the neighborhood of $B^0$, we have
\begin{equation}
    \hspace{-0.02in} \frac{\partial v(B^0)}{\partial b_{ij}} = \frac{\partial \mathcal{L}(x^*; B^0)}{\partial b_{ij}} = (\bar{\mu}_{ij}^*-\ubar{\mu}_{ij}^* + \lambda_{i}^*-\lambda_{j}^*)(\theta_{i}^*-\theta_{j}^*). \label{derivative of optimal cost to susceptance}
\end{equation}

Therefore, the first-order effect of opening line $(i,j)$ on the optimal cost of \eqref{eq:OPF} is given by
\begin{equation}
    \Delta v_{ij} \coloneqq \frac{\partial v(B^0)}{\partial b_{ij}} (0-b_{ij})
    =-(\bar{\mu}_{ij}^* -\ubar{\mu}_{ij}^* + \lambda_{i}^*-\lambda_{j}^*) f_{ij}^*,\label{eq:dv}
\end{equation}
\edit{where $f_{ij}^* \coloneqq b_{ij}(\theta_i^*-\theta_j^*)$.}

\begin{rmk}
    We contrast the sensitivity result \eqref{eq:dv} with that given in \cite{sensitivity}. First, note that an alternative OPF formulation is used in \cite{sensitivity}, which is
    \begin{subequations} \label{eq:OPF2}
    \begin{align} 
        \text{min} \quad & c^\top p^g \label{eq:OPF2:obj} \\
        \text{s.t.} \quad & \mathbbold{1}^\top (p^g - p^d) = 0 \label{eq:OPF2:a} \\
        & \ubar{f} \leq \Psi (p^g - p^d) \leq \bar{f} \label{eq:OPF2:b} \\
        & \ubar{p}^g \leq p^g \leq \bar{p}^g
        \label{eq:OPF2:c}
    \end{align}
    \end{subequations}
    \edit{where $\Psi$ is the $\mathrm{PTDF}$ matrix in \eqref{eq:flow}}.
    
    It is a standard result that the two formulations \eqref{eq:OPF} and \eqref{eq:OPF2} are equivalent in the sense that $(p^g)^*$ is optimal for \eqref{eq:OPF} if and only if it is optimal for \eqref{eq:OPF2}. To see this, notice that $\theta$ in \eqref{eq:OPF:a} is given by $\theta = (B^0)^\dagger (p^g - p^d) + a\mathbbold{1}$ for any $a \in \mathbb{R}$. Reformulating \eqref{eq:OPF:b} in matrix form leads to $\ubar{f} \le D\edit{(A^0)}^\top \theta \le \bar{f}$. Replacing  $\theta$ by $(B^0)^\dagger (p^g - p^d) + a\mathbbold{1}$ in the line flow expression and noting $\mathbbold{1} \in \ker(\edit{(A^0)}^\top)$, we see \eqref{eq:OPF:b} is equivalent to \eqref{eq:OPF2:b}.

    Therefore, we may adapt the sensitivity result in \cite{sensitivity} using the OPF model \eqref{eq:OPF}, which leads to the optimal cost sensitivity metric in \cite{sensitivity} as
    \begin{equation} \label{eq:sensitivity:Ruiz}
        \delta(v) = -(\bar{\mu}_{ij}^* -\ubar{\mu}_{ij}^*) f_{ij}^*.
    \end{equation}

    Comparing \eqref{eq:sensitivity:Ruiz} with the proposed one in \eqref{eq:dv}, we see the difference is the missing $\lambda_{i}^*-\lambda_{j}^*$ term in \eqref{eq:sensitivity:Ruiz}. The advantage of the proposed sensitivity metric will be justified through extensive computational experiments. \qed
\end{rmk}

\subsection{\edit{Sensitivities of Other Parameters}} \label{sect:sensitivity_bs}

\edit{Sensitivities of the optimal cost to other parameters can be derived similarly. In particular, we consider the sensitivities of the optimal cost $v(B^0)$ with respect to relevant parameters under power transfer and line close.}

\edit{First, a power transfer from bus $h$ to $k$ can be modeled as a decrease in $p^d_h$ and an increase in $p^d_k$, whose sensitivity is
\begin{equation}
    \frac{\partial \mathcal{L}(x^*;B^0)}{\partial p^d_k} -\frac{\partial \mathcal{L}(x^*;B^0)}{\partial p^d_h} = \lambda_k^* - \lambda_h^*.
\end{equation}}

\edit{Second, adding a (fictitious) line changes the susceptance matrix only, so the sensitivity of the optimal cost with respect to adding a line between buses $h$ and $k$ is
\begin{equation}
    \frac{\partial \mathcal{L}(x^*;B^0)}{\partial b_{hk}} = (\lambda_h^*-\lambda_k^*)(\theta_h^*-\theta_k^*).
\end{equation}
These sensitivities will be used to estimate the effect of bus split in Section \ref{sect:bsdf}.}

\section{Incorporation of Bus Split and Bus Split Distribution Factor} \label{sect:bsdf}

To incorporate bus split into the sensitivity metric above, we discuss how bus split can be equivalently viewed as a generalized line switching. As a byproduct, we also provide an expression for \edit{the} \emph{bus split distribution factor} ($\mathrm{BSDF}$), a distribution factor quantifying the effect of bus split akin to $\mathrm{PTDF}$ and $\mathrm{LODF}$. 

For ease of exposition we may assume the bus to be split into two busbars is \edit{the last bus (bus $n$)}. The busbars after the split are therefore buses $n$ and $n+1$. Without loss of generality we further assume that the buses are numbered such that after busbar split, bus $n$ is adjacent to buses numbered between $1$ and $m$, and bus $n+1$ is adjacent to buses from $m+1$ to $n-1$. \edit{There are cases where there are parallel lines between buses and we assume that if there are multiple lines between bus $n$ and bus $k$, where $m+1\le k\le n-1$, then after split, all these lines will only connect bus $n+1$ and bus $k$.} We also denote the power injection that busbar $n+1$ carries by $p_{n+1}$ so that bus $n$, after the bus split, has power injection $p_n - p_{n+1}$.

% \textcolor{red}{REMOVE IN THE END}
% \vspace{2in}

After bus split, the DC power flow equations change from $p = B^{0}\mathrm{\theta}$ to \eqref{eq:DCpost}:

\begin{strip}
% \begin{figure*}
\begin{equation} \label{eq:DCpost}
    \begin{bmatrix}
        p_1 \\ \phantom{p} \\ \phantom{p} \\  \vdots \\ \phantom{p} \\ \phantom{p} \\ p_{n-1} \\ p_n - p_{n+1} \\ p_{n+1}
    \end{bmatrix} = 
    \begin{bmatrix}
        \sum_{k=1}^n b_{1k} & & \cdots & & & -b_{1(n-1)} & -b_{1n} & 0 \\
        & & & & & & \vdots & \vdots \\
        \vdots & & \ddots & & & \vdots & -b_{mn} & 0 \\
        & & & & & & 0 & -b_{(m+1)n} \\
        & & & & & & \vdots & \vdots \\
        -b_{(n-1)1} & & \cdots & & & \sum_{k=1}^n b_{(n-1)k} & 0 & -b_{(n-1)n} \\
        -b_{n1} & \cdots & -b_{nm} & 0 & \cdots & 0 & \sum_{k=1}^m b_{nk} & 0 \\
        0 & \cdots & 0 & -b_{n(m+1)} & \cdots & -b_{n(n-1)} & 0 & \sum_{k=m+1}^{n-1} b_{nk}
    \end{bmatrix}
    \begin{bmatrix}
        \theta_1' \\ \phantom{p} \\ \phantom{p} \\  \vdots \\ \phantom{p} \\ \phantom{p} \\ \theta_{n-1}' \\ \theta_n' \\ \theta_{n+1}'
    \end{bmatrix}
\end{equation}
\end{strip}
% \end{figure*}
We notice that the $(n-1)$-th leading principal matrix of the above susceptance matrix is identical to that of $B^0$, while for the $n$-th row and column, the $(m+1)$ to $(n-1)$-th elements shift downward and rightward to the $(n+1)$-th row and column, respectively.

We now perform Kron reduction to eliminate bus $n+1$ for the above system \eqref{eq:DCpost}, which amounts to solving for $\theta_{n+1}'$ using the last row and back substituting it for the remaining rows. Specifically, if we denote $\sum_{k=m+1}^{n-1} b_{nk}$ by $\Sigma$, we can solve for $\theta_{n+1}'$ and represent it as
\begin{equation}\label{theta_n+1}
    \theta_{n+1}' = \frac{p_{n+1}}{\Sigma} + \frac{\sum_{k=m+1}^{n-1} b_{nk}\theta_k'}{\Sigma}. 
\end{equation}
Back substituting $\theta_{n+1}'$ into the first $n$ rows and rearranging the resulting DC power flow model, we have
\begin{equation} \label{eq:newDC}
    \tilde{p}(\theta') = B^0 \theta',
\end{equation}
where $\tilde{p}(\theta') \in \mathbb{R}^n$ is the equivalent injection after bus splitting parametrized by the phase angle vector $\theta'$. The implication about \eqref{eq:newDC} is that bus split can be modeled by changes in bus power injections while keeping the network model intact. 

\subsection{Bus Split Distribution Factor}

The equivalent injection $\tilde{p}(\theta')$ is given by the following expression:
\begin{equation} \label{eq:ptilde}
    \hspace{-0.1in}\tilde{p}_i(\theta') = 
    \begin{cases}
        p_i & \text{if } i \le m \\
        p_i + \frac{b_{ni}}{\Sigma}p_{(n+1)n}(\theta') & \text{if } m+1 \le i \le n-1 \\
        p_n - p_{(n+1)n}(\theta') & \text{if } i = n
    \end{cases}
\end{equation}
where $p_{(n+1)n}(\theta') := p_{n+1} - \sum_{k=m+1}^{n-1} b_{nk} (\theta_n' - \theta_k')$. The notation $p_{(n+1)n}(\theta')$ is indicative of its physical meaning: it is the line flow through the fictitious, zero impedance line between buses $n+1$ and $n$ \emph{before} bus split.

The physical interpretation of the model \eqref{eq:newDC} is now clear: it describes the fact that the bus split is equivalent to power transfers between buses that are adjacent to bus $n+1$ after bus split and bus $n$, the amount of which is proportional to the absolute susceptance between the two buses, and the sum of which is equal to the flow through the fictitious line between buses $n+1$ and $n$ as a result of the equivalent transfers.

The physical interpretation above enables us to find the equivalent transfer $p_{(n+1)n}(\theta')$. First, note that by KCL, the sum of flow between $n+1$ and $n$ and the injection at busbar $n$ (which is $p_n - p_{n+1}$) is equal to the sum of line flows between $n$ and $k$ for $1 \le k \le m$. That is, 
\begin{equation}
    p_{(n+1)n}(\theta') + p_n - p_{n+1} = \sum_{k=1}^m p_{nk}(\theta').
\end{equation}
On the other hand, the additional power flow through the fictitious line $(n+1,n)$ due to the equivalent power transfers can be represented using the \emph{power transfer distribution factors} ($\mathrm{PTDF}$s). In particular, we denote the additional flow on line $(i, j)$ due to per unit transfer from bus $k$ to bus $m$ by $\mathrm{PTDF}_{ij,km}$. It therefore follows that
\begin{multline}
    p_{(n+1)n}^0 + \frac{\sum_{j=1}^{m} \sum_{k=m+1}^{n-1} \mathrm{PTDF}_{nj, kn} b_{nk}}{\Sigma} p_{(n+1)n}(\theta') \\
    + p_n - p_{n+1} = p_{(n+1)n}(\theta') + p_n - p_{n+1},
\end{multline}
where $p_{(n+1)n}^0$ is the original fictitious line flow on line $(n+1,n)$ before bus split. Solving for $p_{(n+1)n}(\theta')$ leads to
\begin{equation} \label{eq:fictitiousp}
    p_{(n+1)n}(\theta') = \frac{p_{(n+1)n}^0 \Sigma}{\Sigma - \sum_{j=1}^{m} \sum_{k=m+1}^{n-1} \mathrm{PTDF}_{nj, kn} b_{nk}}.
\end{equation}

We are now in a position to define $\mathrm{BSDF}$. According to \cite{van2023bus}, the $\mathrm{BSDF}$ is defined as the ratio between the change of line flow due to bus split and the fictitious flow on line $(n+1,n)$ before bus split. Leveraging \eqref{eq:fictitiousp}, we arrive at the following definition of $\mathrm{BSDF}$ for line $(i,j)$ due to the split of bus $n$:
\begin{equation}
    \mathrm{BSDF}_{ij,n} = \frac{ \sum_{k=m+1}^{n-1} \mathrm{PTDF}_{ij, kn} b_{nk}}{\Sigma - \sum_{j=1}^{m} \sum_{k=m+1}^{n-1} \mathrm{PTDF}_{nj, kn} b_{nk}}.
\end{equation}
It can be shown that the definition above is consistent with that in \cite{van2023bus}, but the expression is cleaner and the derivation provides clear physical intuition.

% The equivalent effect of bus split on $B^0$ is encoded in \eqref{eq:ptilde}. For example, the negative line susceptance $b_{kn}$ for line $(k,n)$ reduces by $b_{kn}^2/\Sigma$ for all $m+1 \le k \le n-1$. In addition, the power injection at bus $n$ shifts to buses $m+1$ to $n-1$. For brevity, the detailed expressions on the change of $B^0$ and $p$ will not be given. However, we note that the sensitivity can be similarly derived using envelope theorem since both changes affect system parameters only.

\subsection{Relationship Between Bus Split and Line Switching} \label{sect:bs-ls}

\edit{There is a more standard interpretation of bus split based on network analysis. Instead of retaining the original susceptance matrix $B^0$ as in \eqref{eq:newDC}, we can derive an alternative form of the Kron reduced post-split model leveraging \eqref{theta_n+1} as
\begin{equation}
    p + \Delta p = (B^0 + \Delta B) \theta',
\end{equation}
where $\Delta p$ is
\begin{equation} \label{eq:dp}
    \Delta p_i = 
    \begin{cases}
        0, & \text{if } i \le m \\
        \frac{b_{ni}}{\Sigma}p_{n+1}, & \text{if } m+1 \le i \le n-1 \\
        - p_{n+1} & \text{if } i = n
    \end{cases},
\end{equation}
and $\Delta B$ is
\begin{equation}
    \Delta B_{ij} =
    \begin{cases}
        0, & \text{if } i \le m \text{ or } j \le m \\
        \frac{-b_{in}b_{nj}}{\Sigma}, & \text{if } m + 1 \le i, j \le n-1 \\
        b_{ij}, & \text{if } m + 1 \le i, j \le n, i \text{ xor } j = n \\
        -\Sigma, & \text{if } i=j=n
    \end{cases}
\end{equation}
Note that the $\Delta B$ matrix has the desirable graph Laplacian structure, meaning the diagonal elements are exactly the negative sums of the corresponding off-diagonal ones. Therefore, these perturbation terms model the effect of changing the susceptances of the corresponding lines.}

\edit{The above model indicates that bus split can be equivalently modeled by the following three operations:
\begin{enumerate}
    \item Power injections and withdrawals at bus $n$ and $k$ where $m+1\le k \le n-1$.
    \item Opening lines incident to bus $n+1$.
    \item Connect buses that are adjacent to bus $n+1$, or increase the line susceptances if they are already connected.
\end{enumerate}
Therefore, bus split can be viewed as a combination of line opening, power transfer, and line closing/reinforcement.}

\subsection{First-Order Cost Change Due to Bus Split}

\edit{With the equivalent bus split model in place, we now present the estimated change in the optimal cost of the OPF problem in \eqref{eq:OPF}. According to the discussion in Section \ref{sect:sensitivity_bs}, the optimal cost change due to power transfers can be estimated as:
\begin{align}
     \Delta v_n^1 &=p_{n+1}\lambda_n^* - \sum_{k=m+1}^{n-1} (\frac{b_{kn}}{\Sigma} p_{n+1})\lambda_k^*.
\end{align}
Second, the optimal cost change due to openings of lines that connect with bus $n+1$ can be estimated as:
\begin{align}
    \Delta v_n^2 &= \sum_{k=m+1}^{n-1}-(\bar{\mu}_{nk}^* -\ubar{\mu}_{nk}^* + \lambda_{n}^*-\lambda_{k}^*) f_{nk}^*. \label{bus split cost change}
\end{align}
Third, the estimated optimal cost change due to line connections/reinforcement among buses adjacent to bus $n+1$ is:
\begin{equation}
    \Delta v_n^3 = \sum_{m < i \ne j < n} (\tilde{\mu}_{ij}^* -\utilde{\mu}_{ij}^* + \lambda_{i}^*-\lambda_{j}^*) (\theta_i^* - \theta_j^*) \frac{b_{in}b_{jn}}{\Sigma},
\end{equation}
where $\tilde{\mu}_{ij}^* = \bar{\mu}_{ij}^*$ if $i$ is adjacent to $j$, $\tilde{\mu}_{ij}^* = 0$ otherwise, $\tilde{\mu}_{ij}^*$ is defined similarly.}
% One question is each line can possibly connect with one of the bus after split and here is how we define which lines are connected or not connected with bus $n$ after split: the sensitivity factor of each line can be obtained at first and we will let one line that has the least negative sensitivity value not connect with bus $n$ after split, together with power transfer, which in total is equivalent to bus split. We assume only load power transfer will happen. These are not difficult to achieve, since the connection status among the buses and the power transfer type can be both controlled through switches. 

\edit{Overall, the total estimated optimal cost change due to bus split is calculated as:}
\begin{equation}
    \edit{\Delta v_n = \Delta v_n^1+\Delta v_n^2+\Delta v_n^3.}
    \label{total bus split cost change}
\end{equation}

\section{\edit{One-Step Pivot-Based} Improved Heuristic}

\edit{The sensitivity factors developed in the previous section enable a fast preliminary ranking of candidate line-switching and bus-splitting actions. However, since these factors are based on first-order information, they may not fully capture the change in the optimal cost due to a discrete topology change. An alternative is to resolve the DCOPF \eqref{eq:OPF} for each candidate topology, but this leads to a significant increase in computation time. To strike a balance between accuracy and computational efficiency, we refine the sensitivity-based ranking process by analyzing the effect of each candidate action on the current optimal basis. Specifically, a one-step simplex pivot is applied, which provides a better estimate of the optimal cost change.}

% Consider the optimization problem \eqref{eq:OPF}, \edit{we first estimate the change in the optimal cost associated with each candidate line opening and bus split based on the sensitivity analysis developed in the previous section. Starting from this sensitivity-based estimate, we further refine it by evaluating each candidate through a one-step pivot of the simplex method, which yields a more accurate estimate of the resulting cost change.} In this way, candidate actions can be ranked according to their \edit{estimated cost reduction in an efficient way.}

We \edit{start by considering problem} \eqref{eq:OPF} in standard form:
\begin{subequations} \label{eq:standard}
\begin{align} 
    v = \text{min } &c{^\top}x\label{cost function III}\\
    \text{s.t. } &Ax=b\label{constraints}\\
    & x \ge 0. \label{variables}
\end{align}    
\end{subequations}
\edit{We recall the following standard terminology:} if there is $x_{B}\in \mathbb{R}^{n_B}$ which is a sub-vector of $x$ (where $n_B$ is the number of rows of $A$) such that $[A_{i_{1}}, \dots, A_{i_{n_B}}]\in \mathbb{R}^{n_B\times n_B}$ is a full rank matrix, \edit{where $A_{i_k}$ is the $i_k$-th column of $A$ with $k=1, \dots, n_B$,} then $x_{B}$ is called a basic variable vector \edit{and this full rank matrix is called a basis}. \edit{In this way, $x$ can be partitioned as $\left[\begin{smallmatrix} x_B \\ x_N \end{smallmatrix}\right]$ where $x_N$ are called non-basic variables, with $N$ being the index set of the non-basic variables.} 
% Let $N$ be the non-basic variables index set, \edit{and the optimal solution $x^*$ includes $x_B^*$ and $x_N^*$. Therefore, basis $B$ and the remaining part of $A$ which is called non-basis can both be obtained.}

% First of all, the optimization problem should be formulated in standard form. However, $\theta$ is free, so $\theta = \theta^{+}-\theta^{-}$ is introduced, where $\theta^{+}, \theta^{-} \in \mathbb{R}^n_{\ge 0}$. 

To convert \eqref{eq:OPF} into standard form, we decompose the unrestricted voltage angle variables as the difference between two nonnegative vectors $\theta = \theta^{+}-\theta^{-}$, where $\theta^{+}, \theta^{-} \in \mathbb{R}^n_{\ge 0}$. \edit{We also introduce nonnegative slack variables $\bar{f}^{+}, \ubar{f}^{+}, \bar{p}^{+}$, and $\ubar{p}^{+}$ for the upper and lower line flow and generation limits to convert the corresponding inequality constraints into equality ones. Let $e \coloneqq [1, {-1}]^\top$ and $L \coloneqq ee^\top = \left[\begin{smallmatrix} 1 & -1 \\ -1 & 1 \end{smallmatrix}\right]$, the standard form constraint matrix $A$ and the right-hand-side vector $b$ can be represented compactly as
\begin{equation}
    A = \begin{bmatrix}
        I & \begin{bmatrix}-B^0 & B^0\end{bmatrix} & \mathbbold{0} \\[2mm]
        \begin{bmatrix}\mathbbold{0}\\ e\otimes I\end{bmatrix} & \begin{bmatrix} L\otimes D(A^0)^\top\\\mathbbold{0} \end{bmatrix} & I
    \end{bmatrix}, \;\;\; 
    b = \begin{bmatrix} p^d \\ \bar{f} \\ -\ubar{f} \\ \bar{p}^g \\ -\ubar{p}^g \end{bmatrix},
\end{equation}
where $I$ and $\mathbbold{0}$ are identity and zero matrices of appropriate dimensions, respectively, and $\otimes$ denotes Kronecker product. In addition, the decision vector in standard form is $x = (p^g, \theta^+, \theta^-, \bar{f}^+, \ubar{f}^+, \bar{p}^+, \ubar{p}^+)$. }

% The coefficient matrix is
%  \begin{align}
%     A = [A_{1}|A_{2}]
% \end{align}
% where 
% \begin{align}
%     A_{1} = \begin{bmatrix}
%     I_{n\times n}&-B^{0}&B^{0}\\
%     O_{m\times n}&H_{m\times n}&-H_{m\times n}\\
%     O_{m\times n}&-H_{m\times n}&H_{m\times n}\\
%     I_{n\times n}&O_{n\times n}&O_{n\times n}\\
%     -I_{n\times n}&O_{n\times n}&O_{n\times n}\\
%     \end{bmatrix}
% \end{align}

% \begin{align}
%     A_{2} = \begin{bmatrix}
%     O_{n\times (2n+2m)}\\
%     I_{(2n+2e)\times (2n+2m)}\\
%     \end{bmatrix}
% \end{align}
% and the constant vector is
% \begin{align}
%     b = \begin{bmatrix}
%     p^d\\
%     \bar{f}\\
%     -\ubar{f}\\
%     \bar{p}^g\\
%     -\ubar{p}^g\\
%     \end{bmatrix}
% \end{align}
% where $x=[p^g; \theta^{+}; \theta^{-}; w]$, $p^g\in \mathbb{R}^{n}, \theta^{+}\in \mathbb{R}^{n}, \theta^{-}\in \mathbb{R}^{n}, w\in \mathbb{R}^{(2n+2m)}$.

% Before the derivation of the improved heuristic algorithm, the following is needed:

\edit{We consider both line opening and bus splitting in the heuristic algorithm. For bus split, we will restrict our discussion to the case where the new bus $n+1$ is connected to only one bus after bus split. The results can be easily extended to the more general cases, but they are not considered here due to space limitation. While line opening only modifies the constraint matrix $A$ and leaves the RHS vector $b$ unchanged, bus splitting under the restricted model amounts to a line switching together with an equivalent power transfer, and therefore modifies both $A$ and $b$.} The feasibility and optimality of $x^*$ might not apply to the updated optimization problem. Different computational implications can be drawn depending on how the model is updated.

% \edit{Given the optimal basis of the original DCOPF \eqref{eq:OPF}, a candidate topology control action changes the constraint matrix $A$. Different computational implications can be drawn depending on which columns are updated. We first consider the case when the control action only affects the nonbasic columns.} 

\edit{We first consider the case where only the nonbasic columns are affected by the control action. In this case, the reduced-cost tests are only needed for the modified columns, as the next proposition shows.}

\begin{Proposition} \label{thm:redcostcheck_nonbasic}
Suppose a candidate topology control action modifies only nonbasic columns of the constraint matrix $A$ in \eqref{eq:standard}. Then, to verify whether the current basis remains optimal for the updated problem, it suffices to check the reduced-cost conditions only for the modified nonbasic columns.
\end{Proposition}

\begin{proof}
\edit{For the updated problem, the reduced cost of nonbasic variable $j\in N$ is
\begin{equation}
    \bar c_j' = c_j - c_B^\top B^{-1}A_j'.
\end{equation}
If column $j$ is unchanged, then $A_j'=A_j$, and hence
\begin{equation}
    \bar c_j' = c_j - c_B^\top B^{-1}A_j = \bar c_j \ge 0,
\end{equation}
since the original basis is optimal. Therefore, only the reduced costs corresponding to modified nonbasic columns need to be reevaluated.}
\end{proof}

\edit{In the proposed formulation, opening a line modifies exactly four columns of matrix $A$, due to the updates of the relevant terms in $B^0$ and the associated diagonal entries of $D$. The bus split is equivalent to removing the single connection between the new bus and its only adjacent bus after the split, and thus also leads to the modification of exactly four columns of $A$. Therefore, only four reduced-cost conditions need to be checked in either condition.}

\edit{However, when one of the basic columns is modified, the optimality check becomes more complicated. The fact is stated formally in the next proposition.}

% \begin{Proposition}
% When all changes occur in the non-basic columns, for the optimality condition verification, which is $c_{j}-c_{B}^{\top}B^{-1}A_{j}' \ge 0$, for $j \in N$, four formulas for line switching or bus split need to be checked.
% \end{Proposition}

% \begin{proof} 
% Before the change of line susceptance, an optimal solution is obtained, which is equivalent to $c_{j}-c_{B}^{\top}B^{-1}A_{j} \ge 0$. If the changes only influence non-basic columns, it is known that only $B^{0}$ and $H$ in $A$ will be affected, which correspond to four columns; they are all non-basic columns. The remaining non-basic columns will stay the same, ensuring their reduced costs are the same, which are non-negative. While the reduced costs of the four changed columns will change, so only four reduced costs or formula need to be checked. 
% \end{proof}

% \begin{Proposition}
% When there are changes in the basic columns, all optimality condition verification is required, which is $c_{j}-c_{B}^{\top}B'^{-1}A_{j} \ge 0$, for $j \in N$.
% \end{Proposition}

\begin{Proposition} \label{thm:redcostcheck_basic}
    Suppose a candidate topology control action modifies at least one basic column of the constraint matrix $A$ in \eqref{eq:standard}. Then, in order to verify whether the current basis remains optimal for the updated problem, all reduced-cost conditions need to be reevaluated.
\end{Proposition}

\begin{proof}
    \edit{In the original problem, optimality of the basis implies
    \begin{equation}
        \bar{c}_j = c_j - c_B^\top B^{-1} A_j \ge 0, \quad \forall j \in N.
    \end{equation}
    Now suppose the topology action modifies at least one basic column, so that the basis matrix changes from $B$ to $B'$. Then, for the updated problem, the reduced cost of nonbasic variable $j\in N$ is given by
    \begin{equation}
        \bar c_j' = c_j - c_B^\top (B')^{-1}A_j'.
    \end{equation}
    Even if a nonbasic column remains unchanged, i.e., $A_j' = A_j$, we still have
    \begin{equation}
        \bar c_j' = c_j - c_B^\top (B')^{-1}A_j,
    \end{equation}
    which is generally different from
    \begin{equation}
        \bar c_j = c_j - c_B^\top B^{-1}A_j.
    \end{equation}
    Therefore, once the basis matrix changes, the reduced-cost conditions associated with all nonbasic variables may change. Hence, all reduced-cost conditions must be reevaluated.}
\end{proof}

% \begin{proof}
% Before the change of line susceptance, one optimal solution is obtained, the reduced cost is $c_{j}-c_{B}^{\top}B^{-1}A_{j} \ge 0$ for all $j\in N$. If the changes will affect $B$, all these formula might not hold true anymore.
% \end{proof}

\edit{Propositions \ref{thm:redcostcheck_nonbasic} and \ref{thm:redcostcheck_basic} identify the condition under which the reduced-cost check can be simplified. We next connect the optimality check with the sensitivity metric developed earlier, and show that, when only nonbasic columns are affected by a line switching action, a positive sensitivity in \eqref{derivative of optimal cost to susceptance} implies that at least one of the modified nonbasic variables has negative reduced cost and therefore can enter the basis in a cost-improving one-step pivot.}

\begin{Proposition}
    Consider opening line $\ell=(i,j)$. Suppose this action modifies only nonbasic columns of the constraint matrix $A$ in \eqref{eq:standard}. If the sensitivity of the optimal cost with respect to $b_{ij}$ in \eqref{derivative of optimal cost to susceptance} is positive, then at least one reduced cost associated with the modified nonbasic columns is negative.
\end{Proposition}

\begin{proof}
    \edit{Let $v(b_{ij})$ denote the optimal cost of \eqref{eq:standard} as a function of line susceptance $b_{ij}$. Since $\partial v(b_{ij})/\partial b_{ij} > 0$, there exists a sufficiently small $\epsilon>0$ such that
    \begin{equation} \label{eq:prop3:vcomp}
        v(b_{ij}-\epsilon) < v(b_{ij}),
    \end{equation}
    i.e., the optimal cost reduces after the susceptance change.}
    
    \edit{Suppose, for contradiction, that all reduced costs associated with the modified nonbasic columns remain nonnegative when the susceptance is changed from $b_{ij}$ to $b_{ij}-\epsilon$. Since only nonbasic columns are affected, Proposition \ref{thm:redcostcheck_nonbasic} implies that the current basis remains optimal for the perturbed problem. This would imply that the optimal objective value does not decrease, contradicting \eqref{eq:prop3:vcomp}. Therefore, for at least one modified nonbasic column, the corresponding reduced cost becomes negative at $b_{ij}-\epsilon$.}

    \edit{Now consider such a modified nonbasic column $A_k(b_{ij})$. Since the line susceptance enters the matrix $A$ linearly, the updated column $A_k(b_{ij})$ depends affinely on $b_{ij}$. It follows that the reduced cost $\bar c_k(b_{ij}) = c_k - c_B^\top B^{-1}A_k(b_{ij})$ is also an affine function of $b_{ij}$. Because $\bar c_k(b_{ij}) \ge 0$ before susceptance perturbation and $\bar c_k(b_{ij}-\epsilon) < 0$ after perturbation, this affine function decreases as $b_{ij}$ decreases. Therefore,
    \begin{equation}
        \bar c_k(0) < 0.
    \end{equation}
    Hence, after line $\ell$ is opened, at least one reduced cost associated with the modified nonbasic columns is negative.}
\end{proof}

\edit{The next proposition shows that, even when the basis matrix is modified, the corresponding basic solution can be reevaluated efficiently thanks to the rank-1 nature of the perturbation.}

\begin{Proposition}
\label{thm:basis_update}
\edit{Consider opening line $\ell=(i,j)$. Suppose this action modifies at least one basic column of the constraint matrix $A$ in \eqref{eq:standard}. Let the updated basis matrix be written as
\begin{equation}
    B(\delta)=B+\delta E = B+uv^\top,
\end{equation}
where $\delta=-b_{ij}$. If $1+v^\top B^{-1}u \neq 0$, then the updated basic solution is given by
\begin{equation}
    x_B(\delta)
    =(B(\delta))^{-1}b
    =x_B^*-\frac{B^{-1}uv^\top x_B^*}{1+v^\top B^{-1}u},
\end{equation}
where $x_B^*=B^{-1}b$ is the original basic solution. Moreover, the corresponding basic objective value satisfies
\begin{equation}
    c_B^\top x_B(\delta)
    =c_B^\top x_B^*
    -\frac{c_B^\top B^{-1}uv^\top x_B^*}{1+v^\top B^{-1}u}.
\end{equation}
In particular, if $x_B(\delta)\ge 0$, then the updated basis remains feasible, and the resulting cost change is $c_B^\top x_B(\delta)-c_B^\top x_B^*$.}
\end{Proposition}

\begin{proof}
\edit{Since $B(\delta)=B+uv^\top$ and $1+v^\top B^{-1}u \neq 0$, the Sherman-Morrison formula \cite[Sect. 3.8]{Meyer2000} yields
\begin{equation}
    (B(\delta))^{-1} 
    = B^{-1}-\frac{B^{-1}uv^\top B^{-1}}{1+v^\top B^{-1}u}.
\end{equation}
Multiplying both sides by $b$ gives
\begin{equation}
    x_B(\delta)
    =(B(\delta))^{-1}b
    =B^{-1}b-\frac{B^{-1}uv^\top B^{-1}b}{1+v^\top B^{-1}u}.
\end{equation}
Notice that the original basic solution is $x_B^*=B^{-1}b$, we have
\begin{equation}
    x_B(\delta)
    =x_B^*-\frac{B^{-1}uv^\top x_B^*}{1+v^\top B^{-1}u}.
\end{equation}
Premultiplying by $c_B^\top$ gives
\begin{equation}
    c_B^\top x_B(\delta)
    =c_B^\top x_B^*
    -\frac{c_B^\top B^{-1}uv^\top x_B^*}{1+v^\top B^{-1}u}.
\end{equation}
If $x_B(\delta)\ge 0$, then the updated basic solution is feasible for the modified basis, and therefore the corresponding cost change is $c_B^\top x_B(\delta)-c_B^\top x_B^*$.}
\end{proof}

\edit{With Propositions \ref{thm:redcostcheck_nonbasic}--\ref{thm:basis_update} in place,} the \edit{one-step pivot-based} improved heuristic algorithm is obtained by considering \edit{two cases: the topology action leads to} changes only in the nonbasic columns \edit{of $A$, or} at least one change appears in one of the basic columns.

If the topology change only affects the nonbasic columns of $A$, feasibility of the current basic solution is not affected \edit{in the case of line switching} (as both the basis matrix and the RHS vector stay the same), and we only need to check its optimality. \edit{For bus split, feasibility of the current basic solution also needs to be reexamined with respect to the updated RHS vector.} \edit{For optimality check,} it suffices to examine the reduced costs associated with the modified columns per Proposition \ref{thm:redcostcheck_nonbasic}. If \edit{all such reduced costs remain nonnegative}, the optimal solution \edit{remains unchanged}. Otherwise, a one-step pivot is performed to estimate the updated cost. 

Let 
\begin{equation}
    J \coloneqq \left\{ j\in N: c_{j}-c_{B}^{\top}B^{-1}A_{j}' < 0 \right\}
\end{equation}
be the index set of the modified nonbasic columns whose reduced costs are negative, where $A_j'$ denotes the modified nonbasic column. For each $j \in J$, the simplex ratio test gives:
\begin{equation}
    \alpha_j
    =\min\left\{\frac{(x_B^*)_i}{(B^{-1}A_j')_i}:\; (B^{-1}A_j')_i>0\right\}.
\end{equation}
The corresponding one-step pivot update is
\begin{equation}
    x_B' = x_B^* - \alpha_j B^{-1}A_j', \quad x_j' = \alpha_j,
\end{equation}
and the resulting optimal cost change is estimated by
\begin{equation}
    \Delta\mathrm{cost} = \min \left\{\alpha_j(c_{j}-c_{B}^{\top}B^{-1}A_{j}'):j\in J\right\} \label{non-basis cost change}.
\end{equation}

% Here is the step, $J=\{j\in N: (c_{j}-c_{B}^{\top}B^{-1}(A_{j}+\delta e_{j}) \le 0)\}$ and
% \begin{align}
%     \text{for }&j \in J\notag\\
%     &\delta_{i}=\text{min}\{\frac{(x_B^*)_{i}}{(B^{-1}A_{j}')_{i}}:(B^{-1}A_{j}')_{i}\ge 0\}\\
%     % &i=\text{argmin }_{i: (B^{-1}A_{j}')_{i}\ge 0}\{\frac{x_{B,i}}{(B^{-1}A_{j}')_{i}}\}\notag\\
%     % &\delta_{i} = \frac{x_{B,i}}{(B^{-1}A_{j}')_{i}} \notag\\
%     &x_{B}'=x_{B}^*-\delta_{i}B^{-1}A_{j}'\\
%     &x_{j}'=\delta_{i}
%     % &\Delta\text{cost} =\delta_{i}(c_{j}-c_{B}^{T}B^{-1}A_{j}')\notag
% \end{align}

% The cost change will be selected as:
% \begin{align}
%     \Delta\text{cost} = \text{min}\{\delta_{i}(c_{j}-c_{B}^{\top}B^{-1}A_{j}'):j\in J\} \label{non-basis cost change}
% \end{align}

\edit{We next consider the case where the control action changes} at least one basic column of $A$. \edit{In this case, both} the feasibility and optimality may change. If the susceptance of a line $\ell = (i, j)$ is reduced to zero, the updated basis matrix becomes
\begin{equation}
    B(\delta) = B + \delta E = B + uv^\top,
\end{equation}
where $\delta = -b_{ij}$. By Proposition \ref{thm:basis_update}, if $1 + v^\top B^{-1}u \neq 0$, then the updated basic solution is
\begin{equation}
    x_B(\delta) = x_B^*-\frac{B^{-1}uv^\top x_B^*}{1+v^\top B^{-1}u},
\end{equation}
and the corresponding objective value is
\begin{equation}
    c_B^\top x_B(\delta)
    =c_B^\top x_B^*
    -\frac{c_B^\top B^{-1}uv^\top x_B^*}{1+v^\top B^{-1}u}.
    \label{eq:basis_cost_change}
\end{equation}

If $x_{B}(\delta) \ge 0$, then the updated solution is still feasible, and the associated optimal cost change is estimated by
\begin{equation}
    \Delta \mathrm{cost}
    =c_B^\top x_B(\delta)-c_B^\top x_B^* = -\frac{c_B^\top B^{-1}uv^\top x_B^*}{1+v^\top B^{-1}u}.
\end{equation}

On the other hand, if $x_B(\delta) \not\ge 0$, then the updated basic solution is infeasible, and we need to restore its feasibility before obtaining the optimal cost change estimate. For each nonbasic column $j \in N$, let the search direction $D_j$ be $D_j \coloneqq -B(\delta)^{-1} A_j$. To restore feasibility, we perform the following three operations:

First, we check whether there is an index $i$ such that $(x_B(\delta))_i < 0$ and $(D_j)_i < 0$. If there is, then feasibility cannot be restored by bringing $A_j$ into the basis. 

Second, we determine the smallest step along $D_j$ that eliminates all negative components of $x_B(\delta)$, namely
\begin{equation}
    d_j
    =\max\left\{
    \frac{-(x_B(\delta))_i}{(D_j)_i}:\;
    (x_B(\delta))_i<0,\; (D_j)_i>0
    \right\}.
\end{equation}

Third, we determine the largest step along $D_j$ that would not flip one nonnegative component of $x_B(\delta)$ to negative, namely
\begin{equation}
    \alpha_j
    =\min\left\{
    \frac{-(x_B(\delta))_i}{(D_j)_i}:\;
    (x_B(\delta))_i \ge 0,\; (D_j)_i < 0
    \right\}.
\end{equation}
The nonbasic variable can produce a feasible updated basic solution if and only if $d_j \le \alpha_j$. If no such $j\in N$ exists, then no feasible one-step update can be obtained from the perturbed basis, and the corresponding candidate action is discarded in the improved heuristic. 

Among all nonbasic variables that produce a feasible updated basic solution, regardless of whether $x_B(\delta) \ge 0$, the estimated optimal cost change due to the one-step pivot is
\begin{equation}
    \hspace{-0.04in}\Delta \mathrm{cost}
    =\min\bigl\{
    c_B^\top\bigl(x_B(\delta)+\alpha_jD_j - x_B^*\bigr)+c_j \alpha_j: j\in N
    \bigr\}.
    \label{eq:basic_cost_change_infeasible}
\end{equation}

The same one-step pivot framework applies to the restricted bus split model considered in this paper. Here, the bus split amounts to a single line opening and an equivalent power transfer, and the above derivations carry over accordingly. Let $b'$ be the updated RHS vector considering power transfer. The feasibility condition is $x_B'=B^{-1}b'\ge0$. If the feasibility condition holds and only the nonbasic columns of the constraint matrix are updated, we proceed with a modified version of \eqref{non-basis cost change} by adding $c_B^\top (x_B' - x_B^*)$:
\begin{equation} \label{eq:basic_cost_change_feasible_bs}
    \Delta\mathrm{cost} = \min \left\{\alpha_j(c_{j}-c_{B}^{\top}B^{-1}A_{j}') + c_B^\top (x_B' - x_B^*):j\in J\right\}.
\end{equation}

On the other hand, if the feasibility condition ceases to hold, or when some of the basic columns of $A$ are modified, we proceed with the update procedure leading up to \eqref{eq:basic_cost_change_infeasible}, with the changes that $x_B^*$ is replaced by $(x_B^*)' \coloneqq B^{-1} b'$, and $x_B(\delta)$ replaced by $x_B'(\delta)$ accordingly as its definition depends on $x_B^*$. The estimated optimal cost change becomes:
\begin{equation}
    \hspace{-0.04in}\Delta \mathrm{cost}
    =\min\bigl\{
    c_B^\top\bigl(x_B'(\delta)+\alpha_jD_j - (x_B^*)'\bigr)+c_j \alpha_j: j\in N
    \bigr\}.
    \label{eq:basic_cost_change_infeasible_bs}
\end{equation}

\begin{algorithm}[t]
\caption{Improved heuristic algorithm}
\label{alg:lineoutage}
\begin{algorithmic}[1]
\Require Multipliers $\lambda^*, \ubar{\mu}^*, \bar{\mu}^*, \ubar{\nu}^*, \bar{\nu}^*$, line power flow $f^*$, and number of candidates $T$ .
\Ensure Best candidate action (line switching or bus split)
\State Solve optimization problem \eqref{eq:OPF} to obtain $x^* := ((p^g)^*, \theta^*, f^*, \lambda^*, \ubar{\mu}^*, \bar{\mu}^*, \ubar{\nu}^*, \bar{\nu}^*)$
\For{each switchable line $(i,j)$}
    \State Estimate optimal cost change based on \eqref{eq:dv}
    % \[
    %   \Delta C_{i,j} \approx -f_{i,j}\left(\overline{\mu}_{i,j} - \underline{\mu}_{i,j} + \lambda_i - \lambda_j \right) \hfill \eqref{result of approximated cost change}
    % \]
\EndFor
\For{each splittable bus $n$}
    \State Estimate optimal cost change based on \eqref{total bus split cost change}
    % \[
    %   \Delta C_{i} = \displaystyle\sum_{j \in D_{i}^{-}}\Delta C_{i,j} \hfill \eqref{bus split sensitivity factor}
    % \]
\EndFor
\State Rank lines and buses in an ascending order by their $\Delta v$
\State Choose top $T$ candidates of lines and buses
\For{each candidate \edit{of line switch}}
    \If{Only nonbasic columns of $A$ in \eqref{eq:standard} are modified}
        \State Compute cost change based on \eqref{non-basis cost change}
    \Else
        \State Compute cost change based on \eqref{eq:basic_cost_change_infeasible}
    \EndIf
\EndFor
\edit{
\For{each candidate \edit{of bus split}}
    \If{Only nonbasic columns of $A$ in \eqref{eq:standard} are modified and the updated optimal solution $x' = B^{-1}b'$ is feasible}
        \State Compute cost change based on \eqref{eq:basic_cost_change_feasible_bs}
    \Else
        \State Compute cost change based on \eqref{eq:basic_cost_change_infeasible_bs}
    \EndIf
\EndFor
}
\State Sort candidates in ascending order by updated cost change
\State \Return Action of the first candidate
\end{algorithmic}
\end{algorithm}

\section{Numerical Experiments and Validation}

To validate the quality and computational performance of the proposed improved heuristic algorithm, we benchmark it against existing methods regarding their performance in optimal cost reduction, computational time, and the corresponding action (line switching or bus split). \edit{Six} methods are denoted as \edit{$M0$}, $M1, M2, M3, M4$ and $M5$ respectively, where \edit{$M0$ and }$M1$ correspond to the \edit{price difference switching criterion and }line profits switching criterion described in \cite{ots}, which selects the open line by
\edit{\begin{align}
    &k = \arg \max_{l=(i,j)\in \mathcal{E}} \left\{|\pi_i-\pi_j| \,\big|\, f_{l}(\pi_i-\pi_j) < 0\right\} \label{Line Profits Switching Criterion I}
\end{align}
and}
\begin{align}
    &k = \arg \max_{l=(i,j)\in \mathcal{E}} \left\{|f_{l}(\pi_i-\pi_j)| \,\big|\, f_{l}(\pi_i-\pi_j) < 0\right\} \label{Line Profits Switching Criterion I}
\end{align}
\edit{respectively, }where $\pi$ is the nodal prices. $M2$ is \edit{our} proposed sensitivity-based method considering line switching, which is given with details in \eqref{eq:dv}. $M3$ is also a sensitivity-based method, but it combines bus split and line switching and the sensitivity factor calculation is shown in \eqref{total bus split cost change}, the splittable buses are those buses that connect with at least two lines. $M4$ is the improved heuristic algorithm, which is based on $M3$ after conducting one-step pivot and its details can be found in \textbf{Algorithm 1}. \edit{In the experiments, the number of top candidates $T$ is set to be 6. }$M5$ is the exact solution method \edit{from \cite{zhou2021substation}}, which is given in \eqref{eq:otr}.

Table~\ref{test case} presents the results obtained from applying the \edit{six} methods to nine IEEE test systems, with system sizes up to 13,659 buses. Only one action will be conducted for each method \edit{and N/A means the solution is infeasible after conducting the corresponding action}. \edit{In order to compare our proposed OTR heuristic method where line open is modeled as setting the negative line susceptance to be zero with the method in the literature,} further comparison \edit{among $M0$}, $M1$, and $M2$ is included, where a maximum of 5 lines can be opened, and the final network topology is obtained when there is no feasible line to open. The final optimal cost, running time, and number of open lines \edit{are} shown in Table~\ref{iterated line open}. The analysis is conducted in MATLAB 2024a utilizing the \textsc{Matpower} toolbox. All computations were performed on a desktop computer equipped with an Intel(R) Xeon(R) W-1370 @ 2.90 GHz processor and 32 GB RAM.

\begin{table*}[htbp]
\centering
\caption{Experiment Results of Conducting One Action}
\label{test case}
\renewcommand{\arraystretch}{1.1}
\setlength{\tabcolsep}{4pt}
\resizebox{\textwidth}{!}{
\begin{tabular}{|c|c|ccc|ccc|ccc|}
\hline
\multirow{2}{*}{No.} &
\multirow{2}{*}{Case} &
\multicolumn{3}{c|}{M0} &
\multicolumn{3}{c|}{M1} &
\multicolumn{3}{c|}{M2} \\
\cline{3-11}
 & &
 Cost (\$) & Time (s) & Action
 & Cost (\$) & Time (s)& Action
 & Cost (\$)& Time (s)& Action\\
\hline
1 & 118\_ieee & $1.63e3$ & $0.003$ & line (82, 96) & $1.87e3$ & $0.005$ & line (69,77) & $1.94e3$ & $0.004$ & line (80,96) \\
2 & 1354\_pegase  & N/A & $0.04$ &  line (1074,802) & N/A & $0.04$ & line (1074,802) & $7.31e4$ & $0.04$ & line (256,284)\\
3 & 1888\_rte  & N/A & $0.07$ & line (692,1)& N/A & $0.08$ & line (692,1) & $5.91e4$ & $0.08$ & line (1250,350)\\
4 & 2383\_wp  & $1.40e6$ & $0.11$ & line (310, 6)& $1.32e6$ & $0.12$ & line (31,6) & $1.38e6$ & $0.12$ & line (76,105)\\
5 & 2746\_wp  & $1.16e6$ & $0.15$ & line (387, 223)& $1.16e6$ & $0.16$ & line (2669,199) & $1.16e6$ & $0.17$ & line (2670,2672) \\
6 & 2869\_pegase  & $1.33e5$ & $0.22$ & line (1586,963)& $1.33e5$ & $0.23$ & line (1586,963) & $1.33e5$ & $0.22$ & line (2120,841)\\
7 & 3375\_wp  & $6.42e6$ & $0.25$ & line (812,383)& $6.42e6$ & $0.25$ & line (383, 381) & $6.44e6$ & $0.26$ & line (544,545)\\
8 & 6470\_rte & $9.66e4$ & $1.09$ & line (4644,1) & $9.66e4$ & $1.22$ & line (4644,1) & $9.66e4$ & $1.32$ & line (6146,6219)\\
9 & 13659\_pegase  & $3.82e5$ & $50.53$ & line (7519,4482)& $3.82e5$ & $51.29$ & line (7519,4482) & $3.82e5$ & $46.80$ & line (4543,5236)\\

\hline
\hline

\multirow{2}{*}{No.} &
\multirow{2}{*}{Case} &
\multicolumn{3}{c|}{M3} &
\multicolumn{3}{c|}{M4} &
\multicolumn{3}{c|}{M5} \\
\cline{3-11}
 & &
 Cost (\$) & Time (s) & Action
 & Cost (\$) & Time (s)& Action
 & Cost (\$)& Time (s)& Action\\
\hline
1 & 118\_ieee & $1.52e3$ & $0.01$ & bus 96 & $1.52e3$ & $0.04$ & bus 96 & $1.43e3$ & $4.04$ & bus 82\\
2 & 1354\_pegase & $7.31e4$ & $0.06$ &  line (256, 284) & $7.31e4$ & $0.1$ & bus 256 & $7.31e4$ & $34.58$ & bus 520\\
3 & 1888\_rte & $5.91e4$ & $0.12$ & line (1250, 1242) & $5.91e4$ & $0.22$ & bus 1242 & $5.91e4$ & $24.72$ & bus 1505\\
4 & 2383\_wp & $1.32e6$ & $0.14$ & bus 126 & $1.32e6$ & $0.28$ & bus 126 & $1.31e6$ & $92.99$ & bus 501\\
5 & 2746\_wp  & $1.16e6$ & $0.19$ & bus 462 & $1.16e6$ & $0.38$ & bus 462 & $1.09e6$ & $107.85$ & bus 2454\\
6 & 2869\_pegase & $1.32e5$ & $0.27$ & line (1410, 1040) & $1.32e5$ & $0.48$ & bus 1040 & $1.32e5$ & $144.27$ & bus 414\\
7 & 3375\_wp & $6.42e6$ & $0.31$ & bus 544 & $6.42e6$ & $0.54$ & bus 544 & $6.42e6$ & $181.46$ & bus 136 \\
8 & 6470\_rte & $9.66e4$ & $1.47$ & bus 6162 & $9.66e4$ & $2.10$ & bus 6162 & $9.66e4$ & $221.12$ & bus 6062\\
9 & 13659\_pegase & $3.82e5$ & $48$ & bus 4325 & $3.82e5$ & $50.06$ & bus 4325 & $3.81e5$ & $1571.58$ & bus 3876\\
\hline
\end{tabular}
}
\end{table*}

% \begin{figure}
%     \centering
%     \includegraphics[width=0.85\linewidth]{figures/visualization.png}
%     \caption{Cost and Time Comparison of 2383 Buses System}
%     \label{fig:visualization}
% \end{figure}

\begin{table*}[h]
\centering
\caption{Comparison of Opening Lines Iteratively Among $M0$, $M1$ and $M2$}
\label{iterated line open}
\renewcommand{\arraystretch}{1.1}
\setlength{\tabcolsep}{4pt}

% \begin{tabular}{|c|ccc|ccc|}
% \hline
% \multirow{2}{*}{Case} &
% \multicolumn{3}{c|}{M1} &
% \multicolumn{3}{c|}{M2} \\
% \cline{2-7}
% &
%  Cost (\$)& Time (s) & No.
%  & Cost (\$)& Time (s)& No.
% \\
% \hline
% 118\_ieee & 1.64e3 & 0.02 & 15 & 1.52e3 & 0.02 & 15\\
% 1354\_pegase & 7.31e4 & 0.08 & 2 & 7.31e4 & 0.6 & 15\\
% 1888\_rte & 5.91e4 & 0.001 & 1 & 5.91e4 & 1.18 & 15\\
% 2383\_wp & 1.44e6 & 1.95 & 15 & 1.32e6 & 0.58 & 5\\
% 2746\_wp & 1.26e6 & 2.68 & 15 & 1.26e6 & 2.54 & 15\\
% 2869\_pegase & 1.32e5 & 0.001 & 1 & 1.32e5 & 3.91 & 15\\
% 3375\_wp & 6.37e6 & 5.04 & 15 & 6.37e6 & 7.16 & 15\\
% 6470\_rte & 9.66e4 & 0.001 & 1 & 9.66e4 & 17.91 & 15\\
% 13659\_pegase & 3.81e5 & 0.001 & 1 & 3.81e5 & 301.57 & 6\\
% \hline
% \end{tabular}
% \end{table}

\begin{tabular}{|c|ccc|ccc|ccc|}
\hline
\multirow{2}{*}{Case} &
\multicolumn{3}{c|}{M0} &
\multicolumn{3}{c|}{M1} &
\multicolumn{3}{c|}{M2} \\
\cline{2-10}
&
 Cost (\$)& Time (s) & No.
 & Cost (\$)& Time (s)& No.
 & Cost (\$)& Time (s)& No.
\\
\hline
118\_ieee & $2.16e3$ & $0.64$ & $5$ & $1.64e3$ & $0.003$ & $2$ & $1.52e3$ & $0.48$ & $5$\\
2383\_wp & $1.39e6$ & $0.62$ & $5$ & $1.44e6$ & $0.64$ & $5$ & $1.32e6$ & $0.62$ & $5$ \\
2746\_wp & $1.26e6$ & $0.86$ & $5$ & $1.26e6$ & $0.87$ & $5$ & $1.26e6$ & $0.96$ & $5$ \\
3375\_wp & $6.30e6$ & $1.89$ & $5$ & $6.37e6$ & $1.53$ & $5$ & $6.37e6$ & $1.33$ & $5$ \\
\hline
\end{tabular}
\end{table*}

The experimental results reveal several key observations. \edit{First, for cases 1354\_pegase and 1888\_rte, there is no feasible solution after opening the line according to the criteria of $M0$ and $M1$. While for our proposed methods, feasible solutions can be found.} Additionally, \edit{in all cases,} methods $M3$ and $M4$ demonstrate superior performance among \edit{all }the sensitivity-based approaches, yielding operational costs only marginally higher than those obtained by the exact solution method. Regarding computational efficiency, all sensitivity-based methods exhibit excellent performance, whereas method $M5$ requires substantially longer computation time.

Comparing methods \edit{$M0$, }$M1$ and $M2$, \edit{they all} employ sensitivity-based approaches considering line switching only. \edit{We can see that $M0$ achieves the most cost saving in case 118\_ieee and $M1$ lowers the cost most significantly in case 2383\_wp. However, our proposed line switching heuristic method can always yield a feasible solution. }\edit{We have a further analysis for case 2746\_wp, where the operational cost decreases after line switching and we notice that the generation from generators near the open line has an obvious drop while the generation levels for distant generators remain roughly the same or even increase slightly. This suggests that line switching has a local effect and multiple actions should be considered to achieve additional cost saving.} 

A further comparison \edit{among $M0$,} $M1$ and $M2$ is made, where they can select up to 5 lines to open iteratively by conducting the sensitivity-based method at \edit{some cases}, and the details could be found in Table~\ref{iterated line open}. \edit{Notably}, \edit{$M2$ has the most significant cost decrease}. \edit{Take case 118\_ieee for example, when only one line is allowed to open, our proposed line switching sensitivity method performs worse than $M0$ and $M1$. However, if multiples line can be opened, the proposed method achieves the best performance. Methods $M3$ and $M4$ have optimal results in most cases and near optimal results in some cases for conducting one action and they are not included in the multiple action experiments. The multiple actions experiments are further compared among the OTR heuristic methods.} 

Method $M3$ extends $M2$ by incorporating bus split, achieving improvements of 21.65\% and 4.35\% for cases 118\_ieee and 2383\_wp, respectively, compared to considering line switching alone. This demonstrates that incorporating bus split into network topology reconfiguration significantly enhances system dispatch cost reduction. Comparing methods $M3$ and $M4$, the latter performs a one-step pivot based on the former, resulting in identical costs across all cases, \edit{even though the improved heuristic algorithm $M4$ chooses a different action at cases of 1354\_pegase, 1888\_rte, and 2869\_pegase}. Both methods achieve dispatch costs equivalent to the exact solution method in most cases, despite selecting different switching actions. However, for cases 118\_ieee and 2746\_wp, their optimal costs exceed those of $M5$ by 6.29\% and 6.42\%, respectively. While the one-step pivot operation in $M4$ increases computational time compared to $M3$, both remain within the same order of magnitude. Considering both operational cost and computational efficiency, methods $M3$ and $M4$—the sensitivity-based approach \edit{of OTR}, and the improved heuristic algorithm, respectively—demonstrate the most outstanding performance among sensitivity-based methods across all nine test cases, spanning from small to large-scale systems.

Beyond the per-case comparisons, three broader patterns emerge from the nine IEEE test systems. First, the cost-time trade-off of the proposed sensitivity-based method (M3) and the improved heuristic algorithm (M4) remains favorable as system size grows: for the largest instance (13659\_pegase), both achieve essentially the same operational cost as the exact solution (3.82e5 vs. 3.81e5) while requiring only $\sim$50s compared to 1571.58s for $M5$. Second, the absolute cost gap between $M3/M4$ and $M5$ diminishes with system size, reaching sub-percent levels for the largest case, indicating that the first-order sensitivity evaluation becomes increasingly accurate in large networks. Third, multiple cases exhibit identical costs under different actions (e.g., $M3/M4$ versus $M5$), which empirically validates the theoretical equivalence between bus split and line switching established in the modeling framework: distinct actions can achieve the same optimal objective value. \edit{$M4$ is based on $M3$ by doing one-step pivot, while they achieve the same cost at each case and the former one has a slight increase of computational time. So $M3$ will be selected as the optimal sensitivity-based method among those methods. However, when comparing these with $M5$ which is the exact method considering both line switching and bus split, our proposed method has higher cost in some small and medium test cases, which results from only considering the first order derivative of optimal cost with respect to negative line susceptance.} Collectively, these observations highlight how the innovations, namely modeling line switching via susceptance changes, unifying line switching and bus split within the same sensitivity framework, and employing the computationally efficient sensitivity-based heuristic algorithm rather than costly re-optimization, yield near-exact costs at a fraction of the computation time across systems ranging from small to very large scale.

\section{Conclusion}

In this paper, we propose a sensitivity-based heuristic algorithmic framework that jointly models and optimizes line switching and bus split in an OTR problem. We first derive the sensitivity of line switching by applying the \edit{envelope theorem}, which is then used to estimate the optimal cost change of OPF with respect to line switching, yielding a first-order cost approximation. We then incorporate bus-split operations into the OTR problem and prove that, with respect to their impact on line flows, a bus split can be viewed as a generalized line switching \edit{and some additional operations including power transfer and line connections}. \edit{In doing so, }we derive a simpler form of the BSDF, which quantifies the line-flow changes induced by a bus split, and integrate bus splitting into the sensitivity-based framework. Subsequently, we develop an improved heuristic that, for the top candidate actions selected by the sensitivity metrics, performs a one-step simplex pivot. Finally, the optimal action can be selected from the updated sensitivities.

Comprehensive simulations on nine test cases ranging from 118 to 13,659 buses compare \edit{six} methods in operational cost, computational time, and selected actions. Across all systems, the proposed approaches attain operational costs closest to the exact solution while preserving fast runtimes. Besides, incorporating bus split further improves savings over line switching alone, demonstrating scalability and practical effectiveness on large, realistic grids.

\balance
\bibliographystyle{IEEEtran}
\bibliography{pscc.bib}

\end{document}